\documentclass[lettersize,journal]{IEEEtran}
\usepackage{amsmath}
\usepackage{amssymb}
\usepackage{graphicx}
\usepackage{epstopdf}
\usepackage{graphics}
\usepackage{multicol}
\usepackage{amsmath,epsfig}
\usepackage[ruled,vlined]{algorithm2e}
\DeclareMathOperator{\erf}{erf}
\usepackage{amsfonts}
\usepackage{amssymb}
\usepackage{amsthm}
\usepackage{bbm}
\usepackage{color}
\usepackage{cite}
\usepackage{float}
\usepackage{enumerate}
\usepackage{soul}

\newtheorem{lemma}{Lemma}

\theoremstyle{plain}

\theoremstyle{plain}

\theoremstyle{plain}

\usepackage{mathtools, cuted}

\providecommand{\lemmaname}{Lemma}
\providecommand{\propositionname}{Proposition}
\usepackage{mathtools}
\providecommand{\theoremname}{Theorem}
\providecommand{\lemmaname}{Lemma}
\providecommand{\propositionname}{Proposition}
\providecommand{\theoremname}{Theorem}
\begin{document}

\title{Optimization of Speed and Network Deployment for Reliable V2I Communication in the Presence of Handoffs and Interference}
\author{Haider Shoaib, {\em Student Member, IEEE}, and Hina Tabassum, {\em Senior Member, IEEE} \vspace{-5mm}

\thanks{ H. Shoaib and H.~Tabassum are with the Department of Electrical Engineering and Computer Science, York University, ON, Canada  (e-mail: haider98@my.yorku.ca, hinat@yorku.ca).}
}



\maketitle
\raggedbottom
\begin{abstract}
Vehicle-to-infrastructure (V2I) communication is becoming indispensable for successful roll-out of  connected and autonomous vehicles (CAVs). While  increasing the CAVs' speed improves the average CAV traffic flow, it  increases communication handoffs (HOs)  thus reducing wireless data rates. Furthermore, unplanned density of active base-stations (BSs) may result in severe interference  which negatively impacts CAV data rate.   In this letter, we first  characterize macroscopic traffic flow by considering log-normal distribution of the spacing between CAVs. We then derive novel closed-form expressions for the exact HO-aware rate outage probability and ergodic capacity in a large-scale network with interference.
Then, we formulate a traffic flow maximization problem to optimize the speed of CAVs and deployment density of BSs with HO-aware rate constraints and collision avoidance constraints.
 Our numerical results validate the closed-form analytical expressions, extract useful insights about the optimal speed and BS density, and highlight the key trade-offs between the  HO-aware data rates and CAV traffic flow.
\end{abstract}

\begin{IEEEkeywords}
Connected automated vehicles, vehicular networks, network planning, handoffs, handoff-aware data rate, traffic flow, speed optimization, interference.
\end{IEEEkeywords}

\section{Introduction}
 Vehicle-to-infrastructure (V2I) communication is pivotal to enable  autonomous and well-informed decision-making in connected and autonomous vehicles (CAVs); thus, improving road safety and traffic flow \cite{hobert2015enhancements}. Nevertheless, 
optimizing the speed of CAVs while maximizing the CAV traffic flow and achieving reliable connectivity at the same time is challenging. The reason is that while increasing the CAVs' speed improves the traffic flow,  it also increases communication handoffs (HOs) as the CAVs switch from  one base-station (BS) to another, thus reducing  data rates. Furthermore, unplanned deployment of network infrastructure such as roadside units or BSs may result in severe interference which negatively impacts CAV data rate \cite{9661374}.  A fundamental trade-off thus exists between the communication data rates and CAV traffic flow.  In this context, this letter aims to answer the following questions, i.e., \textbf{(i)} \textit{how to analyze the performance of V2I communication considering HOs and interference?} and \textbf{(ii)} \textit{what is the optimal BS density and average CAV speed that maximize traffic flow?}

Prior research works in the realm of V2I communications have considered mobility in terms of data rates, however traffic flow is typically overlooked. For instance, in \cite{7827020}, Arshad et al. derived HO-aware data rates as a function of the speed of the user device and the HO cost from the BSs. In \cite{9661374}, Hossan et al. presented a stochastic geometry framework to derive the HO-aware coverage probability of a  two-tier wireless network with RF and THz BSs. In \cite{6477064}, Lin et al. proposed a random way-point mobility model to characterize the HO rate and sojourn time using stochastic geometry considering randomly distributed BSs. Recently, in \cite{10001396}, Yan et al. proposed a reinforcement learning approach for joint V2I network selection and autonomous driving policies considering both RF and THz BSs. Their results demonstrated inter-dependency of a CAV's motion dynamics, HOs, and data rate to adopt safe driving behaviours for CAVs.

Another series of research contributions characterizes traffic flow using macroscopic models \cite{11050, SHI2021279}, however V2I communications is typically overlooked. 
 
 None of the aforementioned research works  considered the problem of \textit{CAVs traffic flow maximization} with log-normally distributed CAVs' spacing and \textit{HO-aware data rate} constraints in a large-scale interference-limited wireless network. To this end, our contributions can be summarized as follows:

$\bullet$ We  characterize  macroscopic traffic flow by considering log-normal distribution of the spacing between CAVs. We then derive novel and tractable closed-form expressions for the probability density function (PDF) and cumulative density function  (CDF) of signal-to-interference-plus-noise ratio (SINR), HO-aware rate outage probability and ergodic capacity in a large-scale network with interference. The derived expressions capture the network parameters such as height of the BSs, safety distance of BSs from the road, interference from neighboring BSs, and channel fading. 
   
$\bullet$  We develop a novel optimization framework to jointly optimize the deployment density of the BSs and speed of CAVs to maximize the CAVs' traffic flow with collision avoidance and minimum HO-aware data rate constraints. 
   
  
$\bullet$ Numerical results confirm the accuracy of the derived expressions and extract useful insights related to dynamics of BS density, CAV minimum data rate requirements, average CAV speed, BS heights and safety distances, etc.
\section{System Model and Performance Metrics}

\begin{figure}
    \label{sys model}
    \includegraphics[width=0.4\textwidth]{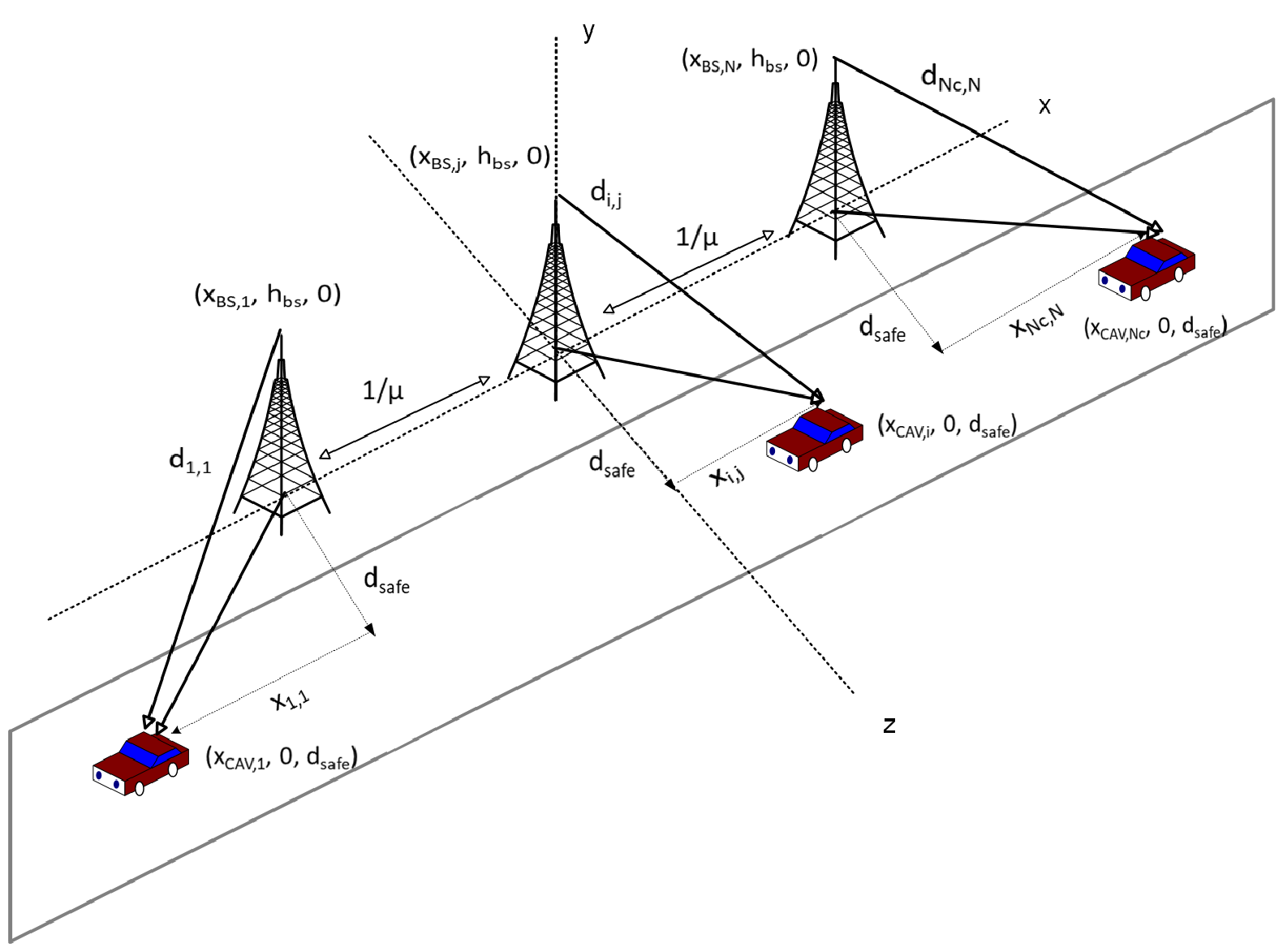}
    \centering
    \caption{Graphical illustration of the V2I communication model for CAVs.}
    \vspace{-5mm}
\end{figure}

We consider a road of length $L_R$ on which $N_c$ CAVs travel. We consider $N$  BSs are deployed alongside the road at a certain distance $d_{\mathrm{safe}}$  with  a density defined as the number of BSs deployed per unit distance, i.e., $\mu = N/L_R$. The distance between BSs is $1/\mu$ as shown in Fig.~1. The CAVs' density $k$ on the road is defined as the number of CAVs per unit distance, and the CAVs' speed is $v$. According to macroscopic traffic flow theory, the flow of vehicles is defined as $q=kv$ vehicles per unit time \cite{MAHNKE20051}. Note that $k$ is inversely proportional to $s$, i.e.,  $k = 1/s$ and $s$ is the spacing between neighboring vehicles. Given the PDF of the density of vehicles $k$ on the road  $f_K(k)$, the  traffic flow can be defined as follows:
\begin{equation}
\label{flow}
    Q=\int_0^{\infty} k v f_K(k) \mathrm{d}k.
\end{equation}
In this paper, we model the inter-vehicle spacing $s$ with a log-normal distribution. This model has been proven to be accurate for daytime hours through various empirical studies \cite{4346442}. The research works observed the traffic flow behaviour during different times of the day, and proved that the inter-vehicle spacing distribution is log-normally distributed during day-time hours (i.e., moderate traffic). Therefore,  we consider the inter-vehicle spacing $s$ to be log-normally distributed during day-time hours with PDF given as:
$$f_S(s) = \frac{1}{s\sigma_{\mathrm{LN}} \sqrt{2\pi}}\exp{\left(-\frac{(\ln{(s)} - \mu_{\mathrm{LN}})^2}{2\sigma_{\mathrm{LN}}^2}\right)},$$ 
where $\mu_{\mathrm{LN}}$ and $\sigma_{\mathrm{LN}}$ are the logarithmic average and scatter parameters of log-normal distribution, respectively, and erf($\cdot$) is the error function. Therefore, \eqref{flow} can be rewritten as follows:
\begin{equation}
    Q=\int_0^{\infty} \frac{1}{s} v f_S(s) \mathrm{d}s = v \exp{\left(\frac{\sigma^2_{\mathrm{LN}} - 2\mu_{\mathrm{LN}}}{2} \right)}.
\end{equation}

Each CAV is assumed to be connected to a single nearest BS at any given time. Considering the distance-based path-loss and short-term multi-path fading at the transmission channel, the received signal power at a given CAV $i$ from a given BS $j$ in the downlink can  be modeled as follows:
\begin{equation}
\label{signal BS}
    S_{i,j} = G_R^{\text{tx}}G_R^{\text{rx}}\left(\frac{c}{4\pi f_R}\right)^2 \frac{P_j^{\text{tx}}}{d_{i,j}^{\alpha}} \chi_{i, j} = \gamma_R P_j^{\text{tx}}d_{i,j}^{-\alpha} \chi_{i, j},
\end{equation}
The signal-to-interference-plus noise ratio (SINR) received at a $i$-th CAV from $j$-th BS can thus be modeled as follows:
\begin{equation}
\label{sinr BS}
    \mathrm{SINR}_{i,j} = \frac{S_{i,j}}{N_R + I_{i}} = \frac{\gamma_R P_j^{\mathrm{tx}}d_{i,j}^{-\alpha} \chi_{i, j} }{N_R + I_{i}},
\end{equation}
where $G_R^{\mathrm{tx}}$ and $G_R^{\mathrm{rx}}$ represent the transmitting and receiving antenna gains, respectively, $P_j^{\mathrm{tx}}$ represents the transmit power of the BS $j$, $\chi_{i, j}$ represents the short-term channel fading of BS $j$ modeled with Rayleigh distribution, $c$ and $f_R$ represent the speed of an electromagnetic wave and RF carrier frequency, respectively, $d_{i,j}$ represents the distance between the $j$-th BS and  $i$-th CAV, and $\alpha$ represents the path-loss exponent. Furthermore, $N_R$ is the thermal noise power at the receiver and $I_{i} = \sum_{k \neq j}P_k^{\mathrm{tx}} \gamma_R d^{-\alpha}_{i,k}\chi_{i,k}$ is the cumulative interference at the $i$-th CAV from the interfering BSs, where $\gamma_R = G_R^{\mathrm{tx}}G_R^{\mathrm{rx}}(c/4\pi f_R)^2$, $d_{i,k}$ represents the distance between the $i$-th CAV and $k$-th interfering BS, and $\chi_{i,k}$ is the power of fading from the $k$-th interfering BS to the CAV.   

Furthermore, as detailed in Fig.~1, the distance between a CAV $i$ and BS $j$ can be calculated using their respective coordinates as 
    $d_{i,j} = \sqrt{x_{i,j}^2 + h_{\mathrm{bs}}^2 + d_{\mathrm{safe}}^2},$
where $h_{\mathrm{bs}}$ represents the height of the BSs, $d_{\mathrm{safe}}$ represents the safety distance from the CAV on the road to the BS, and $x_{i,j}$ is the distance parallel to the road with respect to the location of the CAV on the road to the BS. Subsequently, given the Shannon-Hartley theorem for infinite block-length regime, the data rate without  mobility between a BS and CAV can be defined as
    $R_{i,j} = W\log_2(1 + \text{SINR}_{i,j}) $
where $W$ represents the bandwidth of the channel.

As CAVs drive along the road, HOs between different BSs occur, and due to HO delays and failures, an increase in the rate of HOs can negatively impact the CAV data rate. Therefore, we incorporate the effect of these HOs through the HO-related cost. HO cost is proportional to the HO delay ($h_d$) measured in seconds per HO, and HO rate ($H$) measured in number of HOs per second, i.e.,
   $ H_c = h_d\times H.$
Since we assume nearest BS association for CAVs, we define $H$ as the \textit{number of cell boundaries a CAV crosses per second}, where a cell boundary is the coverage range of a BS. Note that the number of boundaries per unit distance are the same as the BS density $\mu$. By using the speed of CAVs $v$, we can calculate the number of boundaries covered by a CAV per second as $H = \mu v$. Thus, we model the HO-aware data rate \cite{7827020} as:
\begin{equation}
    \label{data rate final}
    M_{i,j} = R_{i,j}(1 -  H_{c, \text{max}}) = R_{i,j}\left(1 - h_d\frac{ \mu v_i}{\mu_{\text{max}} V_{\text{max}}}\right),
\end{equation}
where $H_{c, \text{max}} = h_d\frac{H }{\mu_{\text{max}} V_{\text{max}}} = h_d\frac{ \mu v}{\mu_{\text{max}} V_{\text{max}}}$ is the normalized HO cost in equation \eqref{data rate final} to ensure $0 \leq H_{c, \text{max}} < 1$, $\mu_{\text{max}}$ is the maximum regulated BS density, and $V_{\text{max}}$ is the maximum regulated speed of the CAVs. Finally, a stable CAV connection with the nearest BS requires a minimum HO-aware data rate $R_{\text{th}}$ which ensures that every CAV achieves the required QoS. 

\section{HO-Aware Rate Outage Probability  Analysis}
In this section, we derive novel closed-form expressions for the PDF and CDF of SINR, HO-aware rate outage probability, as well as the ergodic HO-aware data rate. The HO-aware rate outage probability $P_{\mathrm{out}}$ is defined as:
\begin{equation}
\label{outage mobile}
    P_\mathrm{out} = \mathrm{Pr}\left(  M_{i,j} \leq R_{\mathrm{th}} \right) = \mathrm{Pr}\left(Z = \frac{S_{i,j}}{N_R + I_{i}} \leq \gamma_{\mathrm{th}} \right),
\end{equation}
where $\gamma_{\mathrm{th}}$ is the desired SINR threshold given as $\gamma_{\mathrm{th}} = 2^{\frac{R_{\mathrm{th}}}{W(1-H_{c, \mathrm{max}})}}$ and $H_{c, \mathrm{max}} = h_d \frac{\mu v}{\mu_{\mathrm{max}} V_{\mathrm{max}}}$. The outage expression can then be derived as  shown in the following lemma.
\begin{lemma}[HO-aware Rate Outage Probability] Given the PDF and CDF of SINR of $i$th CAV, the closed-form outage expressions can be given as follows:
  \begin{equation}
\label{pout}
     P_\mathrm{out} = 1-\sum^{N}_{k=1, k \neq j} \frac{a_{i,j} e^{\frac{\lambda N_R}{b_{i,k}}}}{b_{i,k} \gamma_{\mathrm{th}} + a_{i,j}}\prod_{l=1, l \neq k}\frac{b_{i,k}}{b_{i,k} - b_{i,l}}.
\end{equation}  
\end{lemma}
\begin{proof}
    To derive the SINR outage, we first calculate the PDF and CDF of SINR. In the sequel, we first determine the  PDF and CDF of $S_{i,j}$ and $I_{i}$. The random variable $S_{i,j}$ is a scaled  exponential random variable, i.e.,  $S_{i,j} = Y = a_{i,j}\chi_{i, j}$ where $a_{i,j} = \gamma_R P_j^{\text{tx}}d_{i,j}^{-\alpha}$. By using a single  variable transformation, the PDF and CDF of $Y$ can be given, respectively, as follows:
\begin{equation}
\label{s pdf}
    f_{Y}(y) = \frac{\lambda}{a_{i,j}}e^{-\frac{\lambda}{a_{i,j}}y}, \quad
    F_{Y}(y) = 1 - e^{-\frac{\lambda}{a_{i,j}}y}.
\end{equation}
The interference  $I_{i} = \sum_{k \neq j}P_k^{\mathrm{tx}} \gamma_R d^{-\alpha}_{i,k}\chi_{i,k} = \sum_{k \neq j}b_{i,k}\chi_{i,k}$ follows  Hypoexponential distribution \cite{hypo}  as $I_i$ is the weighted sum of $n$ independent but non-identical exponential random variables. Each exponential is scaled with a different factor due to the different distance between the CAV and the interfering BSs. {The PDF of the interference can thus be given by:}
\begin{equation}
\label{hypo pdf}
    f_{I_i}(I) = \sum^N_{k=1, k\neq j}\frac{\lambda}{b_{i,k}} e^{-\frac{\lambda}{b_{i,k}} I} \prod_{l = 1, l \neq k}\frac{b_{i,k}}{b_{i,k} - b_{i,l}},
\end{equation}
Now, we define $X = N_R + I_{i}$. 
The PDF  of $X$ can be given after a single random variable transformation as $f_X(x) = f_I(x-N_R )$.
Finally, given the statistics of $X$ and $Y$, we derive the PDF of $Z= Y/X$ \cite{ratio} as:
\begin{align} \label{pdf z equation}
    &f_Z(z) = \int_{0}^{\infty}xf_X(x)f_Y(xz)dx \nonumber \\
    &= \sum^{N}_{k=1, k \neq j} \frac{a_{i,j} \: b_{i,k} e^{\frac{\lambda N_R}{b_{i,k}}}}{(b_{i,k} z + a_{i,j})^2}\prod_{l = 1, l \neq k}\frac{b_{i,k}}{b_{i,k} - b_{i,l}}.
\end{align}
Furthermore, the CDF of $Z$ is as follows:
\begin{equation}
\label{z cdf equation}
     F_Z(z) = 1-\sum^{N}_{k=1, k \neq j} \frac{a_{i,j} e^{\frac{\lambda N_R}{b_{i,k}}}}{b_{i,k} z + a_{i,j}}\prod_{l = 1, l \neq k}\frac{b_{i,k}}{b_{i,k} - b_{i,l}}.
\end{equation}
Finally, the outage in \textbf{Lemma~1} can be calculated by substituting $z=\gamma_{\mathrm{th}}$ in \eqref{z cdf equation}.
\end{proof}
In the following, we characterize the average HO-aware data rate using the statistics of SINR.
\begin{lemma}[Ergodic HO-Aware Data Rate]
  Given the PDF and CDF of SINR, the closed-form ergodic rate expression can be given as follows: 
  \begin{equation}
\label{closed capacity}
M_{\mathrm{avg}}(\mu) = W(1 - \bar{h}_d \mu v ) \sum^{N}_{k=1, k \neq j} \frac{\beta_k a_{i,j} e^{\frac{\lambda N_R}{b_{i,k}}}}{a_{i,j} - b_{i,k}}\prod_{l = 1, l \neq k}\frac{b_{i,k}}{b_{i,k} - b_{i,l}},
\end{equation}
where $\beta_k = \ln(a_{i,j}/b_{i,k}) + \mathrm{atan2}(\lambda/b_{i,k}, 0)$, and $\mathrm{atan2}(y, x)$ is the 2-argument arctangent function.
\end{lemma}
\begin{proof}
We begin by defining the ergodic rate as follows:
    \begin{align}
\label{c2 mean}
    &M_{\mathrm{avg}}(\mu) = \mathbb{E}\left[W \log_2 \left(1 + Z\right)
    (1 - \bar{h}_d \mu v )\right]\nonumber\\
    &= W(1 - \bar{h}_d \mu v ) \mathbb{E}\left[\log_2 \left(1 + Z\right)\right],
\end{align}
where $\bar{h}_d = h_d/(\mu_{\mathrm{max}} V_{\mathrm{max}})$ and $Z$ is a function of $\mu$. Given the definition of ergodic rate \cite{9170614}, we have:
\begin{align}
\label{capacity lemma}
&R_{\mathrm{avg}}(\mu) = \mathbb{E}\left[\log_2 \left(1 + Z\right)\right] = \int^\infty_0 \log_2(1 + Z)f_Z(z) dz, \nonumber \\
&= \frac{1}{\ln(2)}\int^\infty_0 \frac{1 - F_Z(z)}{1+z}dz, \nonumber \\
&= \frac{1}{\ln(2)}\int^\infty_0 \sum^{N}_{k=1, k \neq j} \frac{a_{i,j} e^{\frac{\lambda N_R} {b_{i,k}} } \prod_{l = 1, l \neq k}\frac{b_{i,k}}{b_{i,k} - b_{i,l}}}{(b_{i,k}+a_{i,j})z + b_{i,k} z^2 + a_{i,j}} dz, 
\end{align}
where \eqref{capacity lemma} is a function of $\mu$ because $d_{i,j}, d_{i,k}$ and in turn $b_{i,j}, b_{i,k}$ are  functions of $\mu$. The final step is derived by substituting \eqref{z cdf equation} into \eqref{capacity lemma}. The closed-form ergodic capacity  can be derived by solving the integral as follows:
\begin{equation}
\label{closed capacity}
R_{\mathrm{avg}}(\mu) = \sum^{N}_{k=1, k \neq j} \beta_k\frac{a_{i,j} e^{\frac{\lambda N_R}{b_{i,k}}}}{a_{i,j} - b_{i,k}}\prod_{l = 1, l \neq k}\frac{b_{i,k}}{b_{i,k} - b_{i,l}},
\end{equation}
 Finally, the  HO-aware data rate can be given by substituting  \eqref{closed capacity} into \eqref{c2 mean}.  
\end{proof}

To simplify the optimization, we also provide a worst-case bound on the interference. 
\begin{lemma}[Worst-Case Data Rate at CAV $i$]
    The worst-case signal power will be observed at a CAV when the CAV is located at the halfway point between two BSs (i.e. $\frac{1}{2\mu}$) such that $d_{i,j} = d_{\mathrm{max}} = \sqrt{h_{\mathrm{bs}}^2 + d_{\mathrm{safe}}^2 + \frac{1}{4\mu^2}}$. The worst-case interference will also be observed at this location because the CAV will have maximum signal and interference from the neighbouring BSs such that. $d_{i,k} = \sqrt{h_{\mathrm{bs}}^2 + d_{\mathrm{safe}}^2 + \frac{(2 k+1)^2}{4 \mu ^2}}$. 
    Finally, the ergodic worst-case data rate is given as $$R_{\mathrm{worst}}(\mu) =  \sum^{N}_{k=1, k \neq j} \frac{\beta_k d^{-\alpha}_{\mathrm{max}} e^{\frac{\lambda N_R}{\gamma_R  P^{\mathrm{tx}}_j d^{-\alpha}_{i,k}}}}{d^{-\alpha}_{\mathrm{max}} - d^{-\alpha}_{i,k}}\prod_{l = 1, l \neq k}\frac{d^{-\alpha}_{i,k}}{d^{-\alpha}_{i,k} - d^{-\alpha}_{i,l}}.$$
\end{lemma}
\section{QoS- Constrained Traffic Flow Maximization}
 In this section, we formulate the macroscopic traffic flow maximization problem  with various constraints to jointly optimize the CAV speed $v$ and the BS density $\mu$ in the presence of interference. Note that BS density optimization can also be implemented in practice by dynamically switching the BSs. 
 We present a closed-form expression for the optimal CAV speed and a numerical method to compute  optimal $\mu$.
 The traffic flow maximization problem is formulated as:
\begin{align}
 &(\textbf{P1}) \quad \max_{v,\mu}\qquad Q= v \exp{\left(\frac{\sigma^2_{\mathrm{LN}} - 2\mu_{\mathrm{LN}}}{2} \right)} \nonumber\\
    \mathrm{s.t.} & (\textbf{C1})\quad  v\leq \frac{\exp{\left(\sigma_{\mathrm{LN}} \sqrt{2} \erf^{-1}{\left(2 \epsilon - 1 \right) + \mu_{\mathrm{LN}}} \right)}}{\tau} =V_{\mathrm{safe}}\nonumber\\&
    \textbf{(C2)} \quad v\leq \frac{1}{\bar{h}_d \mu} \left(1 - \frac{R_{\mathrm{th}}}{ W R_{\mathrm{worst}}(\mu)}\right) = V_\mathrm{data}(\mu) \nonumber\\&
    (\textbf{C3}) \quad R_{\mathrm{th}} \leq W R_{\mathrm{worst}}(\mu) \nonumber\\&
    (\textbf{C4}) \quad 0 < v \leq V_{\text{max}} \nonumber\\&
    (\textbf{C5}) \quad 0 \leq \mu \leq \mu_{\text{max}}  \nonumber
\end{align}
where \textbf{C1} is the collision avoidance constraint which ensures that the speed of the CAVs should not exceed $s/\tau$, i.e., \pagebreak
\begin{equation}
\label{avoidance}
    \mathrm{\mathrm{Pr}}\left(v\geq \frac{s}{\tau}\right) = \mathrm{\mathrm{Pr}}\left(s \leq v \tau \right) \leq \epsilon,
\end{equation}
 where $s$ is the distance between two CAVs and is always greater than zero, $\tau$ represents the processing time for the CAVs to act on a decision, and $\epsilon$ is the crash tolerance level. In addition, equation \eqref{avoidance} can be rewritten by substituting the CDF of $s$ where $\mathrm{\mathrm{Pr}}\left(v\geq \frac{s}{\tau}\right) = \frac{1}{2}\left(1 + \mathrm{erf}\left(\frac{\ln{(v\tau)}-\mu_{\mathrm{LN}}}{\sigma_{\mathrm{LN}} \sqrt{2}} \right)\right)$. By using the inverse error function and taking the inverse log of both sides and factoring for $v$, \eqref{avoidance} can be rewritten as in \textbf{C1.}
\textbf{C1} indicates that speed is capped to create an adequate safety distance that is traversed {during the processing time}. We refer to the right side of \textbf{C1} as the maximum safe speed that keeps the crash probability below $\epsilon$. The condition \textbf{C1} shows that the maximum safe speed increases with $\epsilon$, but decreases with the processing time. 

Furthermore, \textbf{C2} is the worst-case HO-aware data rate constraint of the CAV based on  \textbf{Lemma~3}, i.e.,
\begin{equation}\label{M}
    M_{\mathrm{worst}} = W(1 - \bar{h}_d \mu v ) R_{\mathrm{worst}} (\mu)
    \geq R_{\text{th}}.
\end{equation}
By factoring out for $v$ we can rewrite \eqref{M} as in \textbf{C2},
where we refer to the right side of \textbf{C2} as the maximum  speed, denoted by $V_\textrm{data}(\mu)$ that ensures the minimum data rate requirement of the CAVs. 
To ensure that \textbf{C2} does not attain negative values of speed and the problem remains feasible, we introduce the  constraint \textbf{C3}. That is,
\textbf{C3} ensures $\frac{R_{\text{th}}}{ W R_{\mathrm{worst}}(\mu)} \leq 1$. The final two constraints \textbf{C4} and \textbf{C5} cap the speed and BS density to a maximum, respectively. Note that problem \textbf{P1} is a non-linear programming problem due to constraints \textbf{C2} and \textbf{C3} which are a non-linear function of $\mu$. To solve this problem, we first compute the optimal speed $v^*$ which is then further maximized to optimize BS density as  in the following Lemma.
\begin{lemma}
Given that $Q$ is linearly increasing with $v$, and that $v$ is bounded from $V_\mathrm{max}$, $V_\mathrm{safe}$, and $V_\mathrm{data}$, the optimal speed $v^*(\mu)$ is derived as the largest feasible speed that does not violate the three constraints, such that
\begin{equation}
\label{vopt prob}
    v^*(\mu)=\min\{V_\mathrm{max}, V_\mathrm{safe}, V_\mathrm{data}(\mu)\},
\end{equation}
The optimal BS density $\mu^*$ can then be computed  using the \texttt{fminbnd} function in MATLAB which is based on golden-section search algorithm (GSS){\cite{10.5555/1403886}. The GSS method can find the global maximum or minimum of a unimodal function; whereas, it converges to  a local maximum or minimum for a function containing multiple extrema \cite{GILLI2019229}.  GSS is a one-dimensional search that works by reducing the interval in a golden ratio range, where the minimum of the interval lies within the interval. In our case, we want to determine the optimal value of $\mu$ which maximizes $V_{\mathrm{data}}$, so we provide fmindbnd function with the negative of $V_{\mathrm{data}}$ as the objective function. 
The algorithm has a computational complexity of $\mathcal{O}(\log n)$}  \cite{10.1145/367487.367496}. If $\mu^* > \mu_{\mathrm{max}}$, we have $\mu^* = \mu_{\mathrm{max}}$. Also, if $\mu^*$ violates $R_{\mathrm{th}} \leq W R_{\mathrm{worst}}(\mu^*)$, the problem becomes infeasible.
\end{lemma}

It is important to note that due to the interference expression and its dependency on BS density, there is no closed-form expression for $\mu^*$. Therefore, we utilize one-dimensional numerical optimization techniques to solve for $\mu^*$.

Finally, by substituting \eqref{vopt prob} and $\mu^*$ into the objective function of \textbf{P1}, we derive the optimal traffic flow as follows:
 \begin{equation*}
 \label{Q opt prob}
    Q=\min\{V_\mathrm{max}, V_\mathrm{safe}, V_\mathrm{data}(\mu^*)\} \exp{\left(\frac{\sigma^2_{\mathrm{LN}} - 2\mu_{\mathrm{LN}}}{2} \right)}.
\end{equation*}

\section{Numerical Results and Discussions}
In this section, we validate the accuracy of the derived expressions through computer simulations. Furthermore, we demonstrate the sensitivity of optimal traffic flow by changing key parameters such as  crash tolerance level, data rate thresholds, interference, etc.  Unless stated otherwise, the values of the system parameters {\cite{9661374}} are used in the following figures are listed herein. $V_{\mathrm{max}} = 30$ m/s, $\mu_{\mathrm{max}} = 0.01$ BSs/m, $h_d = 3$ s/HO, $R_{\mathrm{th}} = 60$ Mbps, $W = 40$ MHz, $G_R^{\mathrm{tx}} = 1$ dB, $G_R^{\mathrm{rx}} = 1$ dB, $c = 3\times 10^{8}$ m/s, $f_R = 2.1$ GHz, $N_R = 1.507 \times 10^{-13}$ W/$\mathrm{m}^2$, $P_j^{\mathrm{tx}} = 1$ W, $\alpha = 3$, $\lambda = 1$, $\epsilon = 1\%$, $\tau = 6\times 10^{-3}$ sec, $\mu_{\mathrm{LN}} = 0$, $\sigma_{\mathrm{LN}}=1$, $d_{\mathrm{safe}} = 5$ m, $h_{\mathrm{bs}} = 8$ m, $L_R = 2000$ m. 

Fig.~\ref{outage plot mobility} and Fig.~\ref{ergodic capacity plot} demonstrate the outage probability \eqref{pout} and ergodic capacity \eqref{closed capacity} as  a function of $\mu$ for various CAV speeds. The analytical expression match perfectly  with the simulation results. When $\mu$ increases, the outage  increases and rate decreases due to increasing interference and HOs. Furthermore, higher speeds result in higher probability of outage and lower ergodic capacity when compared to lower speeds, which is due to more frequent HOs.  
\begin{figure*}
\vspace{-3mm}
\centering
\begin{minipage}{0.32\linewidth}
\includegraphics[scale=0.35]{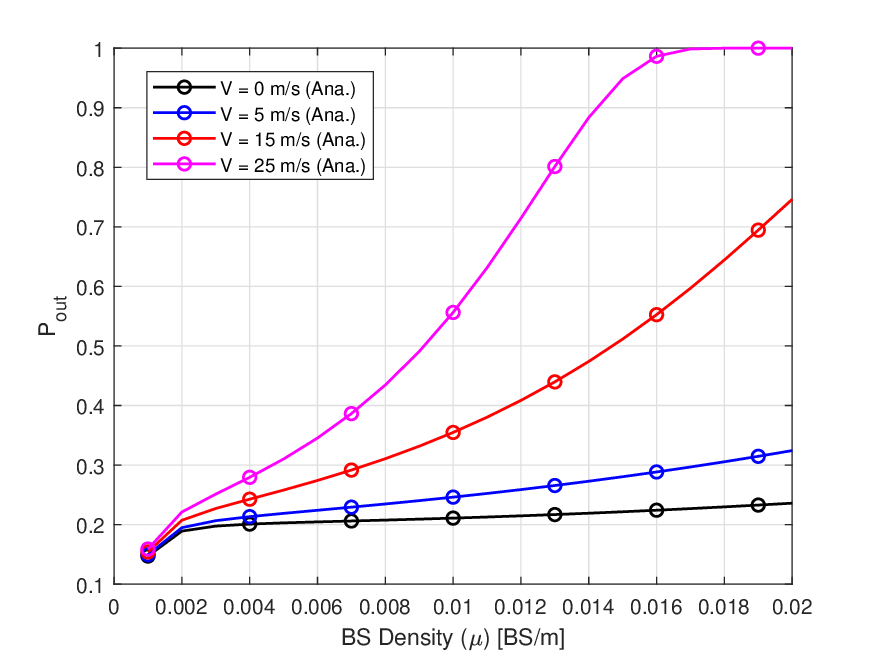}
    \caption{HO-aware rate outage probability as a function of $\mu$  for various speeds and $R_{\mathrm{th}} = 1\times10^8$ bps, $\mu_{\mathrm{max}} = 0.02$ BSs/m.}
    \label{outage plot mobility}
\end{minipage}\hfill
\begin{minipage}{0.32\linewidth}
 \includegraphics[scale=0.35]{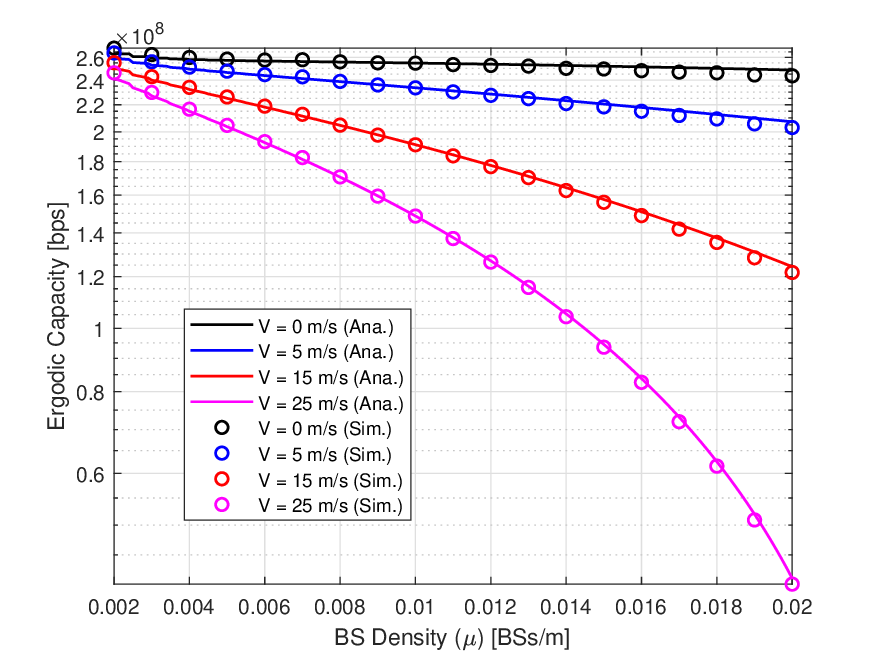}
    \caption{HO-aware capacity as a function of $\mu$ for various speeds for $R_{\mathrm{th}} = 1\times10^8$ bps, $\mu_{\mathrm{max}} = 0.02$ BS/m.}
     \label{ergodic capacity plot}
\end{minipage}\hfill
\begin{minipage}{0.32\linewidth}
 \includegraphics[scale=0.35]{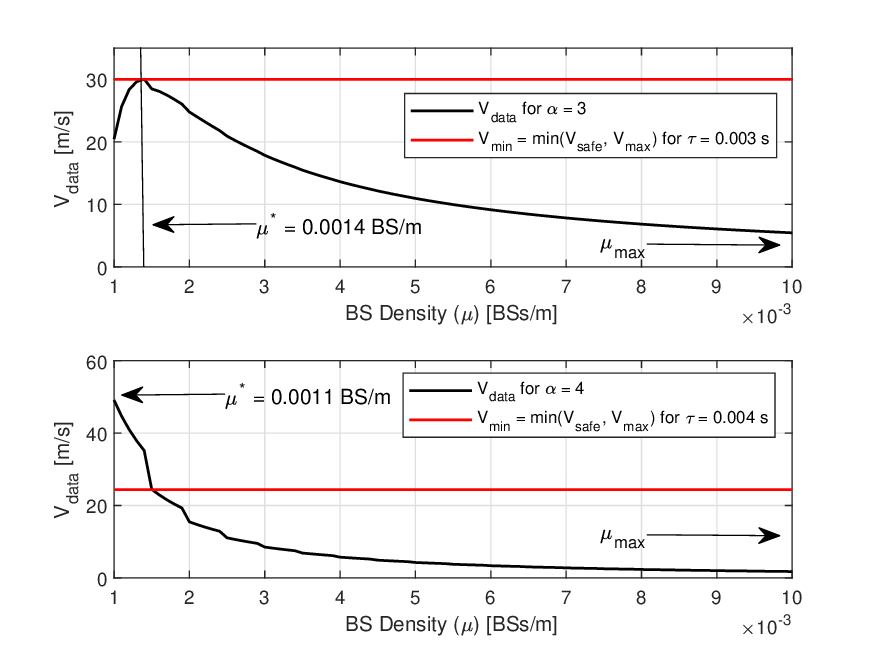}
    \centering
    \caption{ $V_{\mathrm{data}}$ as a function of $\mu$ for $\alpha = 3$ (Top) and $\alpha = 4$ (Bottom) with various CAV processing times $\tau$.}
    \label{vdata}
\end{minipage}
  \vspace{-1.5mm}
\label{fig2}
\end{figure*}

\begin{figure*}
\vspace{-3mm}
\centering
\begin{minipage}{0.3\linewidth}
 \includegraphics[scale=0.35]{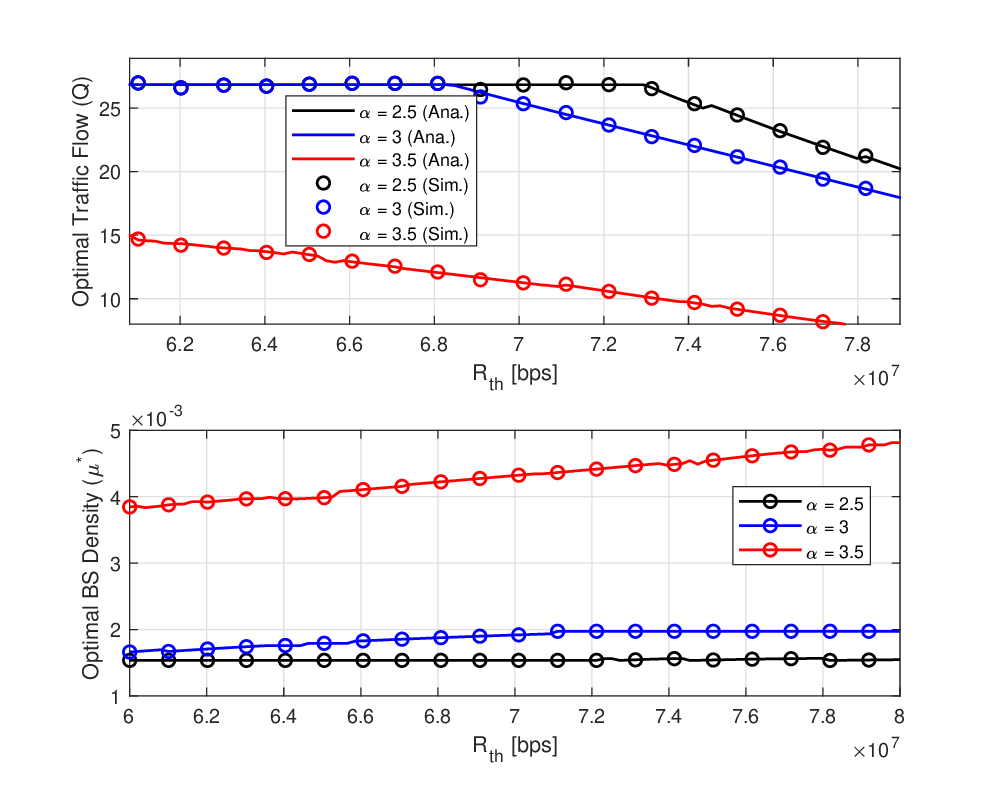}
    \centering
    \caption{Traffic flow $Q$ as a function of $R_{\mathrm{th}}$ with analytical and simulation results (Top) and optimal BS density $\mu^*$ as a function of $R_{\mathrm{th}}$ (Bottom) for various path-loss exponents.}
    \label{qvsrth}
\end{minipage}\hfill
\begin{minipage}{0.32\linewidth}
\includegraphics[scale=0.4]{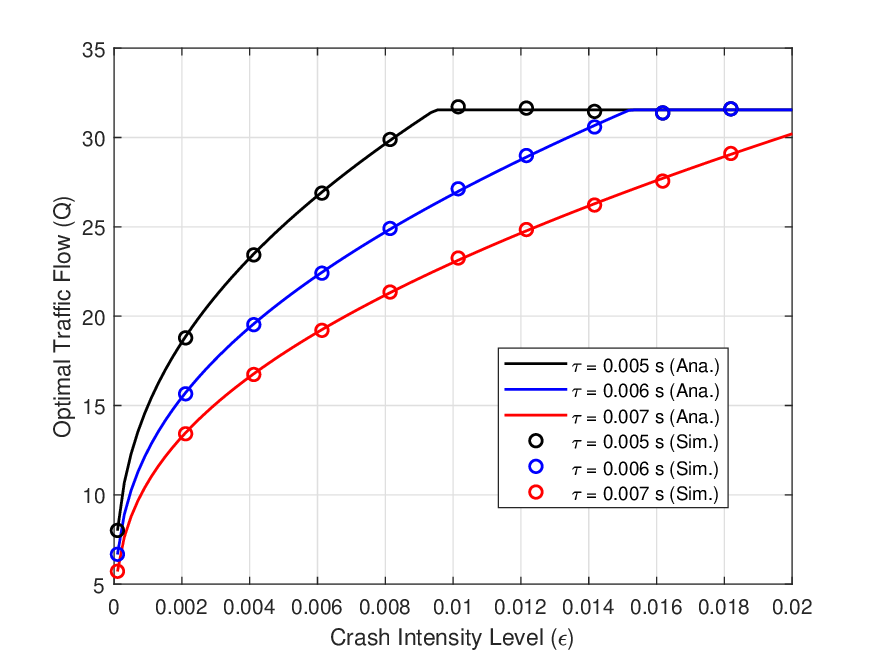}
    \centering
    \caption{Traffic flow $Q$ as a function of crash tolerance  $\epsilon$ and various CAV processing times $\tau$.}
    \label{q_vs_epsilon}
\end{minipage}\hfill
\begin{minipage}{0.33\linewidth}
\includegraphics[scale=0.34]{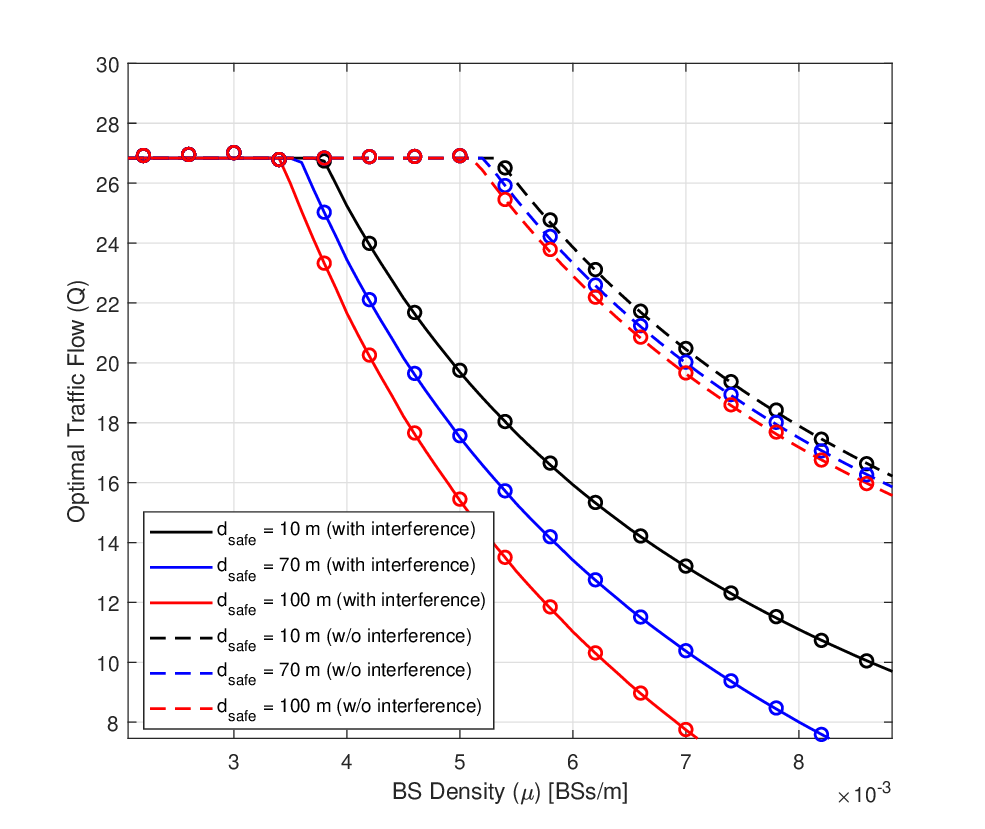}
    \centering
    \caption{Traffic flow $Q$ as a function of $\mu$ with  various BS safety distances with and without interference.}
    \label{qvsmu_dist}
\end{minipage}
\vspace{-5mm}
\label{fig2}
\end{figure*}

Fig.~\ref{vdata} depicts $V_{\mathrm{data}}$ as a function of $\mu$. As $\mu$ increases, $V_{\mathrm{data}}$ with $\alpha = 3$ increases up to a certain point due to vicinity of BSs. However, beyond that point $V_{\mathrm{data}}$  begins to decrease due to an increase in HOs and interference, and the CAVs  need to lower their speed. On the other hand, when $\alpha=4$, the benefit of rate enhancement cannot be seen and a higher $\mu$ only causes more interference and HOs. Furthermore, we show that different CAV processing times $\tau$ varies $V_{\mathrm{safe}}$, where $v^* = V_{\mathrm{safe}}$ for higher $\tau$ values when $V_{\mathrm{safe}} < V_{\mathrm{data}}$.

Fig.~\ref{qvsrth} depicts the impact of the data rate threshold $R_{\mathrm{th}}$ on traffic flow $Q$ and the optimal BS density $\mu^*$ for various path-loss exponents. As seen on the top of  Fig.~\ref{qvsrth}, as $R_{\mathrm{th}}$ increases, $Q$ decreases as CAVs need to reduce their speed to limit the number of HOs to meet the required data rate. For the same reason, when $R_{\mathrm{th}}$, $\mu^*$ increases. The environments with higher path-loss exponents require more dense BS deployment. {Note that we consider 100  points (or $R_{\mathrm{th}}$ values) along the x-axis between  $6\times 10^7$ bps and $8\times 10^7$ bps to compute  analytical traffic flow which requires to compute the optimal density of BSs $\mu^*$ numerically using MATLAB's fminbnd function at each $R_{\mathrm{th}}$ value. Therefore, minor fluctuations are observed due to numerical computations  of $\mu^*$.}

Fig.~\ref{q_vs_epsilon} depicts the relationship between the crash tolerance level $\epsilon$ and optimum traffic flow $Q$ for various CAV processing times. At first, $V_{\mathrm{safe}} <V_{\mathrm{data}} < V_{\mathrm{max}}$, thus the optimal speed is bound to $V_{\mathrm{safe}}$ as in Lemma~4, where $V_{\mathrm{safe}}$ increases as the crash intensity level $\epsilon$ increases since we are relaxing the crash tolerance. After a certain $\epsilon$ value, $V_{\mathrm{safe}} >  V_{\mathrm{data}}$. The optimal speed is bound to $V_{\mathrm{data}}$ as in Lemma~4, which is constant with respect to $\epsilon$ and results in an upbound flat curve. For different CAV processing times $\tau$, the point at which $V_{\mathrm{data}}$ takes over changes. When the CAV processing time is larger, $V_{\mathrm{safe}}$ is smaller since it takes a longer time for the CAV to process decisions compared to smaller $\tau$ values.

Finally, Fig~\ref{qvsmu_dist} depicts the relationship between BS density $\mu$ and traffic flow $Q$ with various BS safety distances. {At first, when the BS density $\mu$ is lower, $V_{\mathrm{safe}} < V_{\mathrm{data}} < V_{\mathrm{max}}$ so the optimal speed is bound to $V_{\mathrm{safe}}$ as in Lemma~4. $V_{\mathrm{safe}}$ is constant with respect to $\mu$ which results in an upbound flat curve. As  $\mu$ increases, HOs and interference increases which lowers $V_{\mathrm{data}}$ to the point where $V_{\mathrm{safe}} > V_{\mathrm{data}}$ and the optimal speed is bound to $V_{\mathrm{data}}$ as in Lemma~4. For different BS safety distances $d_{\mathrm{safe}}$, the switching point between $V_{\mathrm{safe}}$ and $V_{\mathrm{data}}$ changes because the further away the BS is from the CAV, $V_{\mathrm{data}}$ will decrease because of weak signal strength. Furthermore, the traffic flow without interference results in better traffic flow as compared to the traffic flow with interference as the interference deteriorates the data rate, CAV speed, and traffic flow. } 

\section{Conclusion}
This letter presents a framework for V2I communications for CAVs where we derive novel expressions for outage probability and ergodic capacity of HO-aware data rate, and jointly optimize the speed and network deployment in the presence of HOs between BSs and interference due to neighbouring BSs.  Finally, we  demonstrate the trade-off between achievable wireless data rates and traffic flow.
 
\bibliographystyle{IEEEtran}
\bibliography{references}
\vfill
\end{document}